\documentclass[10pt]{amsart}

\usepackage{amssymb,amsthm,amsmath,enumerate}
\usepackage[numbers,sort&compress]{natbib}
\usepackage{color}
\usepackage{graphicx}
\usepackage{amssymb,amsthm,amsmath}
\usepackage[numbers,sort&compress]{natbib}
\usepackage{color}
\usepackage{graphicx}

\title[ ]{  Criteria  for    eigenvalues embedded into the absolutely continuous spectrum  of   perturbed Stark type operators   }

\author{Wencai Liu}
\address[Wencai Liu]{Department of Mathematics, University of California, Irvine, California 92697-3875, USA}\email{liuwencai1226@gmail.com}
\theoremstyle{plain}
\newtheorem{theorem}{Theorem}[section]
\newtheorem{corollary}[theorem]{Corollary}
\newtheorem{lemma}[theorem]{Lemma}
\newtheorem{proposition}[theorem]{Proposition}

\newcommand{\R}{\mathbb{R}}
\newcommand{\Z}{\mathbb{Z}}
\theoremstyle{definition}

\newtheorem{remark}[theorem]{Remark}

\begin{document}


\begin{abstract}
In this paper, we consider  the  perturbed Stark operator
\begin{equation*}
  Hu=H_0u+qu=-u^{\prime\prime}-xu+qu,
\end{equation*}
where $q$ is  the  power-decaying    perturbation.
The    criteria for   $q$ such that $H=H_0+q$ has  at most one  eigenvalue (finitely many, infinitely many eigenvalues) are obtained.
 All the results are quantitative and  are  generalized to the  perturbed Stark type operator.

\end{abstract}
\maketitle
\section{Introduction}

The Stark operator $H_0$ is a self-adjoint operator on $L^2(\R)$ given by the potential $v(x)=-x$:
\begin{equation}\label{Gu}
    H_0u=-u^{\prime\prime}-xu.
\end{equation}
The operator describes a charged quantum particle in a constant electric field.
The Stark effect (named after   Johannes Stark \footnote{It was independently discovered   by the  physicist Antonino Lo Surdo.})
 originates from the interaction between a charge distribution (atom or molecule) and an external electric field.
  In many cases, the particle is also subjected  to an additional electric potential $q$.  For example, the hydrogen Stark effect is governed by an additional
  Coulomb potential.
Stark effect is a important subject in quantum theory, classical electrostatics or other physical  literatures  \cite{solem1997variations,epstein1926stark,courtney1995classical,py1}.
In mathematics, it  also attracts a lot of attentions, see \cite{Kor,Kor1,Kor2,Graf1,Her2,Her1,Her3,Her4,Her5,Ya1,Ya2,Jen1,Jen2,Jen3}.

In this paper,
 we will investigate a class of more general operators, which is called   the   {\it Stark type} operator.
Given    $v_{\alpha}(x)=-x^{\alpha}$  for $x\geq 0$ with  $0<\alpha<2$,  let $\widetilde{v}_{\alpha}(x)$ be an  extension of $v_{\alpha}(x)$ to $\R$ such that $\lim_{x\to -\infty}\widetilde{v}_{\alpha}(x)=\infty$.
Let $H_0^{\alpha}=-D^2 +v_{\alpha}$, which is defined  on $\R^+$ with some boundary condition at $x=0$,  and  $\widetilde{H}_0^{\alpha}=-D^2 +\widetilde{v}_{\alpha}$, which is defined on $\R$.
We call   $H_0^{\alpha}$($\widetilde{H}_0^{\alpha}$) the   Stark type operator.

 The perturbed     Stark  type operator is given by an additional potential:
\begin{equation}\label{Gpwstark}
  H^{\alpha}u=H_0^{\alpha}u+qu( \text{ or } \widetilde{H}^{\alpha}u=\widetilde{H}_0^{\alpha}u+qu),
\end{equation}
where $H_0^{\alpha}=-D^2+v_{\alpha}$ (or $\widetilde{H}_0^{\alpha}=-D^2+\widetilde{v}_{\alpha}$) and  $q$ is the decaying  perturbation.

It is well known that   $\sigma _{\rm ess}(H^{\alpha}_0)=\sigma _{\rm ac}(H^{\alpha}_0)=\R$ and $H^{\alpha}_0$ does not have any eigenvalue.
We are interested in the criteria    of perturbation such that  the associated perturbed Stark type operator has single embedded eigenvalue, finitely many embedded  eigenvalues or infinitely many embedded eigenvalues.
We refer the readers to  \cite{Her1,Her2} and references therein for embedded eigenvalues (resonances) of  operators with  Stark effect.
 In the following,
we always assume $\lim_{x\to\infty}|q(x)|=0$ and $\lim_{x\to-\infty}|q(x)|=0$.
%

For the single embedded eigenvalue problem,
Vakulenko showed if $q(x)=\frac{O(1)}{1+|x|^{\frac{1}{2}+\epsilon}}$ for some $\epsilon>0$, then the perturbed Stark operator $Hu=-u^{\prime\prime}-xu+qu$ has no eigenvalues
in $L^2(\R)$\cite{vakulenko1986nonexistence}.  Naboko and Pushnitskii proved that the perturbed    Stark type operator $H^{\alpha}_0+q$ on $\R^+$ has no eigenvalues if
$|q(x)|\leq \frac{C}{1+x^{1-\frac{\alpha}{2}}}$  with $C< 1-\frac{\alpha}{2}$ \cite{naboko1}. Is the bound $1-\frac{\alpha}{2}$ sharp? If it is not, what is the sharp bound?

Before answer  those questions, we want to mention   the problem of embedded eigenvalues for the  perturbed free    Schr\"odinger  operator $-D^2+V$.
Let $a=\limsup_{x\to \infty} x|V(x)|$.
By a result of Kato \cite{kato}, there is no eigenvalue $E$ with $E>a^2$, which holds  for Schr\"odinger operators in any dimension.
From the classical Wigner-von Neumann  type functions
\begin{equation*}
    V(x)=\frac{c}{1+x}\sin(  kx+\phi),
\end{equation*}
we know that one can not do better than $\frac{a^2}{4}$.
For one dimensional case, Atkinson and Everitt  \cite{atk} obtained the optimal bound $\frac{4a^2}{\pi ^2}$, that is
there is no eigenvalue in  $(\frac{4a^2}{\pi ^2},\infty)$ and that there are examples
with eigenvalues approaching  this bound.  We refer the readers to Simon's paper  for  full  history \cite{simon2017tosio} and  a short note \cite{liu} for a complete proof.

It is natural to ask what will happen at the transition line $\frac{4a^2}{\pi ^2}$.  Transition line is always hard to deal with since it can not be addressed in both sides.

 The  first purpose of this paper is to obtain the sharp spectral transition for existence of  eigenvalues  for      perturbed   Stark type operators
 and also explore what happens in the transition lines for both Schr\"odinger operators and Stark type operators. We should mention that some sharp results for  preservation of   the absolutely continuous spectrum are obtained  \cite{christ2003absolutely,kiselev2000absolutely,killipimrn}.


\begin{theorem}\label{Mainthm1}
Suppose the potential $q$  satisfies
\begin{equation*}
   \limsup_{x\to \infty}{x}^{1-\frac{\alpha}{2}}|q(x)|<\frac{2-\alpha}{4}\pi.
 \end{equation*}
 Then $ -u^{\prime\prime}-x^{\alpha}u+qu=Eu$ admits no $L^2(\R^+)$ solutions for any $E\in \R$.
\end{theorem}
The following result shows that the bound $\frac{2-\alpha}{4}\pi$  is optimal and can be achieved.
\begin{theorem}\label{Mainthm2}
For any $E\in \R$, $a\geq\frac{2-\alpha}{4}\pi$ and $\theta\in [0,\pi]$, there exist potentials $q(x)$ on $\R^+$ such that
\begin{equation*}
   \limsup_{x\to \infty}{x}^{1-\frac{\alpha}{2}}|q(x)|=a,
 \end{equation*}
 and  eigen-equation  $ -u^{\prime\prime}-x^{\alpha}u+qu=Eu$  has  an  $L^2(\R^+)$ solution with the boundary condition $\frac{u^{\prime}(0)}{u(0)}=\tan\theta$.

\end{theorem}
During the proof of Theorem \ref{Mainthm2}, we also proved the following   theorem, which covered the (missing) critical case for the Schr\"odinger operator.
\begin{theorem}\label{Schrcritical}
 For   each pair $(\lambda,a)$ such that $\lambda=  \frac{4a^2}{\pi^2}>0$ and any $\theta\in[0,\pi]$,  there exist      potentials  $V$
such that $\limsup_{x\to \infty}x|V(x)|=|a| $  and the  associated  Schr\"odinger equation   $-u^{\prime\prime}+Vu=\lambda u$  has an $L^2(\R^+)$ solution with the boundary condition $\frac{u^{\prime}(0)}{u(0)}=\tan\theta$.
 \end{theorem}
 \begin{remark}
 With some   modifications in our constructions, we can make the potentials in Theorems  \ref{Mainthm2} and \ref{Schrcritical} smooth.
 \end{remark}

 The proof of Theorems \ref{Mainthm1}, \ref{Mainthm2} and \ref{Schrcritical} is inspired by the Schr\"odinger case (see \cite{atk,eastham1982schrodinger}).
 The novelty here is that  instead of using  sign type functions $V(x)=\frac{c}{1+x}\text{ sgn } (  \sin (kx+\phi))$ during the constructions,
 we use sign type functions piecewisely, namely $V_n(x)=\frac{c_n}{1+x}\text{ sgn } (  \sin (kx+\phi))\chi_{[a_n,b_n]}$ and glue them together.
 This piecewise construction   allows us to address the transition line as we mentioned before.

For the sharp   transition of eigenvalues embedded into $(-2,2)$ for the discrete Schr\"{o}dinger operator, see \cite{jl3}.

Define $P\subset \R$ as
 \begin{equation*}
    P=\{E\in\R: -u^{\prime\prime}-x^{\alpha}u+qu=Eu \text{ has  an }  L^2(\R^+) \text{ solution}\} .
 \end{equation*}
$P$ is the collections of    the eigenvalues of $H_0^{\alpha}+q$ with all the possible boundary conditions at $0$.

The next question is that    what is the criterion   for   
finitely many embedded eigenvalues.  We obtain
\begin{theorem}\label{Maintheoremapr3}
Suppose potential $q$  satisfies
\begin{equation*}
   \limsup_{x\to \infty}{x}^{1-\frac{\alpha}{2}}|q(x)|=a.
 \end{equation*}
 Then we have $$\# P\leq \frac{2a^2}{(2-\alpha)^2}.$$
\end{theorem}

\begin{theorem}\label{Mainthm3}
For any  $\{ E_j\}_{j=1}^N\subset \R$ and $\{\theta_j\}_{j=1}^N\subset [0,\pi]$,
there exist  functions  $q\in C^{\infty}[0,+\infty)$   such that
\begin{equation}\label{Ggoalb}
 \limsup_{x\to \infty}   {x}^{1-\frac{\alpha}{2}} |q(x)|\leq (2-\alpha)e^{2\sqrt{\ln N}}N,
\end{equation}
and for each $E_j$, $j=1,2,\cdots,N$,
the eigen-equation $-u^{\prime\prime}-x^{\alpha}u+qu=E_ju$  has  an  $L^2(\R^+)$ solution with the boundary condition $\frac{u^{\prime}(0)}{u(0)}=\tan\theta_j$.
\end{theorem}

Since $e^{2\sqrt{\ln N}}$ is asymptotically smaller than  $N^{\epsilon}$ with any $\epsilon>0$ as $N$ goes to infinity, we have
\begin{corollary}
For any  $\{ E_j\}_{j=1}^N\subset \R$ and $\{\theta_j\}_{j=1}^N\subset [0,\pi]$,
there exist  functions  $q\in C^{\infty}[0,+\infty)$   such that
\begin{equation*}
 \limsup_{x\to \infty}   {x}^{1-\frac{\alpha}{2}} |q(x)|\leq C(\epsilon)N^{1+\epsilon},
\end{equation*}
 and for each $E_j$, $j=1,2,\cdots,N$,
the eigen-equation $-u^{\prime\prime}-x^{\alpha}u+qu=E_ju$  has  an  $L^2(\R^+)$ solution with the boundary condition $\frac{u^{\prime}(0)}{u(0)}=\tan\theta_j$.
\end{corollary}
\begin{remark}\label{Reop}
\begin{itemize}
\item In our forthcoming paper, we will prove that the bound in Theorem \ref{Maintheoremapr3}  is sharp \cite{liusharpstark}.
\item Theorems \ref{Maintheoremapr3} and  \ref{Mainthm3} implies $O(1)$ is the criterion  for finitely many  $L^2(\R^+)$ for the  perturbed   Stark type operator.
For Schr\"odinger operator, the answer is no. Suppose the  positive sequence $\{E_j\}$ satisfies $\sum_{j}\sqrt{E_j}<\infty$,  Simon  \cite{simdense}  constructed potential $V(x)=\frac{O(1)}{1+x}$ such that $-D^2+V$ has eigenvalues $\{E_j\}$.
\item In \cite{simdense}, the sum of  Wigner-von Neumann  type function  $\sum_{j=1}^N c\frac{\sin(\sqrt{E_j}x+\phi_j)}{1+x}$ is used to create positive eigenvalues $\{E_j\}$ for
the perturbed free Schr\"odinger operator. However, by Liouville transformation, the perturbed   Stark type operator always has ``eigenvalue" 1. So it is hard to use  Wigner-von Neumann  type functions directly to do the constructions in our situations.
\item  Our   proof  in Theorems \ref{Maintheoremapr3} and  \ref{Mainthm3} is  effective. Instead of $O(1)$, the explicit  bounds $(2-\alpha)e^{2\sqrt{\ln N}}N$ and  $\frac{2a^2}{(2-\alpha)^2}$   are  obtained.
\item Our constructions are very general. We only give the  parameters   specific  values in the last step.
\end{itemize}

\end{remark}

The proof of Theorem \ref{Maintheoremapr3} is motivated by \cite{kiselev1998modified}. However the technics are more difficult.
The key idea of  \cite{kiselev1998modified} is to show the almost orthogonality of $\frac{\theta(x,E_1)}{1+x}$ and $\frac{\theta(x,E_2)}{1+x}$ in Hilbert space $ L^2([0,B],(1+x)dx)$ for all large $B$, where $\theta(x,E_1)$ ($\theta(x,E_2)$) is the  Pr\"ufer angle with respect to energy $E_1$ ($E_2$). However, the perturbed  Stark (type) operator has its own difficulty.
By Liouville transformation, we can transfer the eigen-equation $Hu=Eu$ of the  perturbed   Stark (type) operator to the eigen-equation $(-D^2+V)u=u$ of   perturbed free Schr\"odinger operator.  This means that all the new eigen-equations of   perturbed free Schr\"odinger operator shares the common eigenvalue $1$ but with different potentials.
It is very challenged  to deal with the common eigenvalues  for the  perturbed free Schr\"odinger operator or common quasimomentum  for  the perturbed periodic  Schr\"odinger operator since the resonance phenomenon will show up. This is the reason why  the assumption  of nonresonance is needed, for example   \cite{ld1,krs}.
In this paper, we  overcame  the difficulty by two new ingredients.  Firstly, we give a general estimates for the oscillating functions, which is a generalization of Wigner-von Neumann type functions. See the comments right after Lemma \ref{Keyle1}.   We mention  that the possiblely embedded eigenvalues for such (or more general) potentials can be determined (e.g. \cite{Luk13,Luk14}).
Secondly, we  take    the second leading term of the evolution  of the Pr\"ufer angle (the first leading term is 1)  into consideration so  that the   resonant    phenomenon
can be well studied.   Moreover,  the two Theorems for  oscillatory integral and almost orthogonality are universal.  See Section \ref{UOA}.

For the construction part, let us say more.
For the perturbed   Stark (type) operator, under the rational independence assumption of  set $\{E_j\},
$ Naboko and Pushnitskii \cite{naboko1} proved Theorems \ref{Mainthm3} and \ref{Mainthm4} without quantitative bounds.
There are more results for the  free perturbed Schr\"odinger case $H=-D^2+V$.  Naboko \cite{nabdense} and Simon  \cite{simdense} constructed    potentials    for which the associated Schr\"odinger operator has  given eigenvalues with or without rational dependence assumption.
By Pr\"ufer transformation or generalized Pr\"ufer transformation, there are a lot of authors considering the (dense) eigenvalues embedded into the essential spectrum or
absolute continuous spectrum for the perturbed free Schr\"odinger operator, the  perturbed periodic Schr\"odiger operator or  the  discrete Schr\"odiger operator
\cite{krs,lukdcmv,kru,remling2000bounds}.

 Recently, by the combination of Pr\"ufer transformation  (generalized Pr\"ufer transformation) and piecewise potentials,
Jitomirskaya-Liu and Liu-Ong constructed asymptotically flat (hyperbolic) manifolds, perturbed periodic operators and perturbed Jacobi operators with finitely or countable many embedded eigenvalues \cite{ld1,jl,jl3}.
 Here, we  develop the piecewise potential technics in  \cite{ld1,jl,jl3} in several aspects.
 Firstly, we gave the  universal constructions in an  effective   way.
  We gave the  single piece  constructions and also obtained the effective bounds  in Section \ref{ESPC}.
In Section \ref{UGC}, the universal gluing constructions are given and the effective bounds are obtained too. 
 Secondly, as we mentioned before,  we   dealt   with the  resonant eigenvalues situations.
 Unlike the free perturbed or periodic perturbed Schr\"odinger operator, the perturbation is a very small since $q(x)=o(1)$ as $x\to \infty$, the Stark effect $x^{\alpha}$ is much larger than perturbation $q$ in this paper. It is even  hard to imagine  that under suitable constructions, the perturbation will dominate the evolution of the equation.  It turns to be that all the leading entries of all  dominations corresponding to energies are the same.
 So we need to tackle the second domination to distinguish the eigen-equations among different energies. This is the same difficulty in establishing the almost orthogonality among  Pr\"ufer angles.
After overcoming those difficulties,  we are able to prove  Theorem \ref{Mainthm3}. Moreover,  we can construct potentials with infinitely many eigenvalues. See Theorem \ref{Mainthm4} below.
We believe that our method has wide  applications in studying the Schr\"odinger operators.


\begin{theorem}\label{Mainthm4}

Let $h(x)>0$ be   any
function  on  $(0,\infty)$  with $  \lim_{x\to   \infty}h(x) = \infty$ and any sequence $\{\theta_j\}\subset [0,\pi]$.
Then for any given   $\{ E_j\}_{j=1}^{\infty}\subset \R$,  there exist  functions $q\in C^{\infty}[0,+\infty)$  such that
\begin{equation}\label{Ggoala}
    |q(x)|\leq \frac{h(x)}{1+{x}^{1-\frac{\alpha}{2}}}\quad \text{for } x>0,
\end{equation}
and for any $E_j$,
 the eigen-equation $-u^{\prime\prime}-x^{\alpha}u+qu=E_ju$  has  an  $L^2(\R^+)$ solution with boundary condition $\frac{u^{\prime}(0)}{u(0)}=\tan\theta_j$.
\end{theorem}

Now we consider  operators $H_0^{\alpha}=-D^2 +v_{\alpha}$ on $\R^+$ with some fixed boundary condition at $x=0$  (or operators  $\widetilde{H}_0^{\alpha}=-D^2 +\widetilde{v}_{\alpha}$ on $\R$).
In such setting, $E$ is an eigenvalues for $H_0^{\alpha}$ (or $\widetilde{H}_0^{\alpha}$) if and only if $H_0^{\alpha}u=Eu$ has an $L^2(\R^+)$ solution (or $L^2(\R)$).
Based on the previous Theorems and some addition arguments, we have plenty of Corollaries.
\begin{corollary}\label{cor1}
Suppose the potential $q$  satisfies
\begin{equation*}
   \limsup_{x\to \infty}{x}^{1-\frac{\alpha}{2}}|q(x)|<\frac{2-\alpha}{4}\pi.
 \end{equation*}
 Then $H^{\alpha}=H_0^{\alpha}+q$( $\widetilde{H}^{\alpha}=\widetilde{H}_0^{\alpha}+q$) admits no eigenvalues.
\end{corollary}
\begin{corollary}\label{cor2}
For any $E\in \R$, $a\geq\frac{2-\alpha}{4}\pi$, there exist potentials $q(x)$ such that
\begin{equation*}
   \limsup_{x\to \infty}{x}^{1-\frac{\alpha}{2}}|q(x)|=a,
 \end{equation*}
 and  $H^{\alpha}u=H_0^{\alpha}+q$( $\widetilde{H}^{\alpha}=\widetilde{H}_0^{\alpha}+q$)   has an  eigenvalue $E$.

\end{corollary}

\begin{corollary}\label{cor3}
Suppose potential $q$  satisfies
\begin{equation*}
   \limsup_{x\to \infty}{x}^{1-\frac{\alpha}{2}}|q(x)|=a.
 \end{equation*}
 Then the total number of eigenvalues of   $H^{\alpha}=H_0^{\alpha}+q$( $\widetilde{H}^{\alpha}=\widetilde{H}_0^{\alpha}+q$)   is less than $\frac{2a^2}{(2-\alpha)^2}$.
\end{corollary}

\begin{corollary}\label{cor4}
For any  $\{ E_j\}_{j=1}^N\subset \R$,
there exist  functions  $q\in C^{\infty}[0,+\infty)$ $(q\in C^{\infty}(\R))$ such that
\begin{equation*}
 \limsup_{x\to \infty}   {x}^{1-\frac{\alpha}{2}} |q(x)|\leq (2-\alpha)e^{2\sqrt{\ln N}}N,
\end{equation*}
 and  $H^{\alpha}=H_0^{\alpha}+q$ $(\widetilde{H}^{\alpha}=\widetilde{H}_0^{\alpha}+q)$   has eigenvalues
$\{ E_j\}_{j=1}^N$.
\end{corollary}

\begin{corollary}\label{cor5}
Let $h(x)>0$ be   any
function  on  $(0,\infty)$  with $  \lim_{x\to   \infty}h(x) = \infty$.
Then for any given   $\{ E_j\}_{j=1}^{\infty}\subset \R$,  there exist  functions $q\in C^{\infty}[0,+\infty)$ $(q\in C^{\infty}(\R))$ such that
\begin{equation*}
    |q(x)|\leq \frac{h(x)}{1+{x}^{1-\frac{\alpha}{2}}}\quad \text{for } x>0,
\end{equation*}
  and   $H^{\alpha}=H_0^{\alpha}+q$ $(\widetilde{H}^{\alpha}=\widetilde{H}_0^{\alpha}+q)$   has eigenvalues
$\{ E_j\}_{j=1}^{\infty}$.
\end{corollary}
Our paper is organized in the following way.
In the first 8 Sections, we only give the proof of the case   $\alpha=1$ in all the theorems.
In Section \ref{Section:Small}, we will give some preparations.
In Section \ref{UOA}, we will set up universal oscillatory integral and also prove the almost orthogonality between different Pr\"ufer angles. In Section \ref{SE},
we will prove Theorems \ref{Mainthm1} and \ref{Mainthm2}.  In Section \ref{finitelymany}, we will prove  Theorem \ref{Maintheoremapr3}.
In Section \ref{ESPC}, we will give the construction of potentials for single piece and also the effective bounds.
In Section \ref{UGC}, we will give the general method to glue the piecewise functions together and also  the effective bounds.
In Section \ref{Twoapp}, as two applications, we will prove  Theorems \ref{Mainthm3}  and  \ref{Mainthm4}, and as well as all the Corollaries.
In  Section \ref{General}, we will point out the modifications  so that our arguments work for general $\alpha$ with $\alpha\in(0,2)$.

\section{Preparations}\label{Section:Small}

 Let $v $  be a positive function on $\R^+$   and consider the Schr\"odinger equation on $\R^+$,
 \begin{equation}\label{Gsch}
 -u^{\prime\prime}(x)-v(x)u(x)+q(x)u(x)=Eu(x).
 \end{equation}
The  Liouville transformation (see \cite{christ2003absolutely,naboko1}) is given by
\begin{equation}\label{GLiou}
  \xi(x)=\int_0^x\sqrt{v(t)} dt, \phi(\xi)=v(x(\xi))^{\frac{1}{4}}u(x(\xi)).
\end{equation}
We  define a weight function $p(\xi)$ by
\begin{equation}\label{GWei}
  p(\xi)= \frac{1}{v(x(\xi))}.
\end{equation}
We also define a potential   by
\begin{equation}\label{GPQ}
 Q(\xi,E)= -\frac{5}{16}\frac{|v^{\prime}(x(\xi))|^2}{v(x(\xi))^3}+\frac{1}{4}\frac{v^{\prime\prime}(x(\xi))}{v(x(\xi))^2} +\frac{q(x(\xi))-E}{v(x(\xi))}.
\end{equation}

In the following,  we will use the Liouville transformation to perform our proof.
We suppose to take  $v(x)=x^{\alpha}$ for some $0<\alpha<2$ in the following arguments.
As we aforementioned,  we   prove the case for $\alpha=1$ in the beginning, that is $v(x)=x$ for $x\geq 0$. 
 Let $c=(\frac{3}{2})^{\frac{2}{3}}$.
 Under such assumption, one has
 \begin{equation}\label{GLiou1}
  \xi=\frac{2}{3}x^{\frac{3}{2}}, \phi(\xi,E)=(\frac{3}{2})^{\frac{1}{6}}\xi^{\frac{1}{6}}u(x(\xi)),
 \end{equation}
 \begin{equation}\label{GWei1}
  p(\xi)= \frac{1}{c\xi^{\frac{2}{3}}},
\end{equation}
and
\begin{equation}\label{GPQ1}
 Q(\xi,E)= -\frac{5}{36\xi^2}+\frac{q(c\xi^{\frac{2}{3}})-E}{c\xi^{\frac{2}{3}}}.
\end{equation}
Notice that  the potential $Q(\xi,E)$ depends on $q$ and $E$.

Suppose $u\in L^2(\R^+)$  is a solution of \eqref{Gsch} with $v(x)=x$. It follows that
$\phi$ satisfies
\begin{equation}\label{Gschxi}
  -\frac{d^2\phi}{d\xi^2}+Q(\xi,E)\phi=\phi,
\end{equation}
and $\phi\in L^2(\R ^+,p(\xi)d\xi)$. This leads to
$\phi^{\prime}\in L^2(\R ^+,p(\xi)d\xi)$ \cite[Lemma 1]{naboko1}.

Let us introduce the Pr\"{u}fer tranformation.
Let
\begin{equation}\label{GPruf1}
 \phi(\xi,E)=R(\xi,E)\sin\theta(\xi,E),
\end{equation}
and
\begin{equation}\label{GPruf}
\frac{d \phi(\xi,E)}{d \xi}=R(\xi,E)\cos\theta(\xi,E).
\end{equation}
Thus we have
\begin{equation}\label{GPrufRmar14}
 \frac{d\log R(\xi,E)}{d\xi}=\frac{1}{2}Q(\xi,E)\sin2\theta(\xi,E)
\end{equation}
and
\begin{equation}\label{GPrufeTmar14}
  \frac{d\theta(\xi,E)}{d\xi}=1-Q(\xi,E)\sin^2\theta(\xi,E).
\end{equation}

We need one more lemma. See the Appendix for the proof.
\begin{lemma}\label{Keyle2}
Suppose $\lim _{x\to -\infty}\widetilde{q}(x)=\infty$. Let us consider equation
\begin{equation}\label{Gnesc}
  -y^{\prime\prime}+\widetilde{q}(x)y=0.
\end{equation}
Then for any $M>0$,  there is a solution of \eqref{Gnesc}  and $x_0<0$ such that
\begin{equation}\label{Gnesc1}
  |y(x)| \leq e^{-M|x|}
\end{equation}
for   $x<x_0$.
\end{lemma}

\section{Oscillatory integral and Almost orthogonality}\label{UOA}
\begin{lemma}\label{Keyle1}
Let $\beta_1 >0, \beta_2>0$ and $\gamma \neq 0$ be constants. Suppose $\beta_1+\beta_2>1$ and $\beta_2>\frac{1}{2}$.
Suppose $\theta(x)$ is a solution of the following equation on $x>1$,
\begin{equation}\label{Gfourier9}
  \frac{d\theta (x)}{dx}=\gamma+\frac{O(1)}{x^{\beta_1}}
\end{equation}
Let $\beta=\min\{\beta_2,\beta_1+\beta_2-1,2\beta_2-1\}$.
Then for any $1<a<b$, we have
\begin{equation}\label{Gtheta1}
  \int_{a}^b \frac{\sin \theta(x)}{x^{\beta_2}}dx=O(\frac{1}{a^{\beta}}), \int_{a}^b \frac{\cos \theta(x)}{x^{\beta_2}}dx=O(\frac{1}{a^{\beta}})
\end{equation}
and
\begin{equation}\label{Gtheta2}
  \int_{a}^b\frac{|\sin 2\theta(x)|}{x }dx=\frac{2}{\pi}\ln \frac{b}{a}+\frac{O(1)}{a^{\beta_1}}.
\end{equation}

\end{lemma}
\begin{proof}
We only give the proof of \eqref{Gtheta1}. The proof of \eqref{Gtheta2} is similar.
We can assume $a$ is large enough and $\gamma>0$.

 Let $i_0$ be the largest   integer such that $2\pi i_0<\theta(a)$.
 By \eqref{Gfourier9},
 there exist
  $x_0<x_1<x_2<\cdots<x_t<x_{t+1}$  such that $x$ lies in $[x_{t-1},x_t)$ and
 \begin{equation}\label{tildetheta}
    {\theta} (x_i)= 2\pi i_0+ i\pi
 \end{equation}
 for $i=1,2,\cdots,t,t+1$.

By \eqref{Gfourier9}, one has
 \begin{equation*}
    x_{i+1}-x_{i}=\frac{\pi}{ \gamma}+ \frac{O(1)}{x_i^{\beta_1}},
 \end{equation*}
and
 \begin{equation}\label{e2}
  x_i\geq x_0+\frac{i\pi}{2\gamma}.
 \end{equation}
 Similarly, for $x\in[x_i,x_{i+1})$, we have
\begin{equation*}
   {\theta} (x)=2\pi i_0+i \pi +\gamma(x-x_i)+\frac{O(1)}{  x_i^{\beta_1}}.
\end{equation*}
Thus, one has
 \begin{eqnarray}
   \nonumber && \int_{x_i}^{x_{i+1}}|\sin\theta(x) |dx\\
     &=&\int_{0}^{\frac{\pi}{\gamma}}\sin( \gamma x)dx+ \frac{O(1)}{1+x_i^{\beta_1}} =\frac{2}{\gamma}+ \frac{O(1)}{x_i^{\beta_1}}. \label{e1}
 \end{eqnarray}
Notice that $\sin\theta(x) $ changes the sign at $x_i$. The integral also has some cancellation between $(x_{i-1},x_i)$ and $(x_{i },x_{i+1})$.
Let $t^\prime\in\{t,t+1\}$ such that ${t}^\prime $ is odd.

By  \eqref{e1}, we obtain

 \begin{eqnarray}
  \nonumber\int_{a}^b\frac{\sin\theta(x)}{x^{\beta_2}}dx&=&\frac{O(1)}{a^{\beta_2}}+
   \int_{x_1}^{x_{t^\prime}}\frac{\sin\theta(x)}{x^{\beta_2}}dx \\
  \nonumber   &=&\frac{O(1)}{a^{\beta_2}}+O(1)\sum_{i=1}^{t+1}(\frac{1}{x_i^{\beta_2}}-\frac{1}{x_{i+1}^{\beta_2}})+\sum_{i=1}^{t+1}\frac{O(1)}{x_i^{\beta_1}}\frac{1}{x_i^{\beta_2}}\\
   &=&\frac{O(1)}{a^{\beta_2}}+O(1)\sum_{i=1}^{t+1}( \frac{1}{x_i^{2\beta_2}}+\frac{1}{x_i^{\beta_1+\beta_2}})\\
      &=&\frac{O(1)}{a^{\beta_2}}+\frac{O(1)}{a^{2\beta_2-1}}+\frac{O(1)}{a^{\beta_1+\beta_2-1}},\label{ellestimate}
 \end{eqnarray}
where the last equality holds by \eqref{e2}.  By the same argument, we have $\int_{a}^b \frac{\cos \theta(x)}{x^{\beta_2}}dx=O(\frac{1}{a^{\beta}})$. This  completes our proof.
\end{proof}

If we let $\beta_1=\infty$ in Lemma \ref{Keyle1},  \eqref{Gtheta1} reduces to the case of   the Wigner-von Neumann type  functions, which has been proved by plenty of authors. See  \cite{atkinson1954asymptotic,harris1975asymptotic}  for example.
The case that $\beta_2=1$ and $a=1$ in \eqref{Gtheta1} has been established in \cite{ld1}.

 Let
 \begin{equation}\label{Gvapr}
  \frac{q(c\xi^{\frac{2}{3}})}{c\xi^{\frac{2}{3}}}=V(\xi), \text{  for }\xi>0.
\end{equation}
\begin{lemma}\label{Keyleboud}
Suppose $V(\xi)$ in \eqref{Gvapr} satisfies $V(\xi)=\frac{O(1)}{1+\xi}$. Suppose $E_1\neq E_2$.
  Then the  following estimate holds for $\xi>\xi_0>1$
  \begin{equation*}
    \int_{\xi_0}^{\xi} \frac{\sin2\theta(x,E_1)
 \sin2\theta(x,E_2) }{  1+x} dx= \frac{O(1)}{\xi_0^{\frac{1}{3}}}.
 \end{equation*}
  \end{lemma}
  \begin{proof}
  It suffices to prove
   \begin{equation*}
    \int_{\xi_0}^{\infty} \frac{\sin2\theta(\xi,E_1))
 \sin2\theta(\xi,E_2) }{  1+\xi} d\xi=\frac{O(1)}{\xi_0^{\frac{1}{3}}}.
 \end{equation*}
 Observe that by basic trigonometry,
\begin{equation*}
2\sin2\theta(\xi,E_1)\sin2\theta(\xi,E_2)=\cos (2\theta(\xi,E_1)-2\theta(\xi,E_2))-\cos (2\theta(\xi,E_1)+2\theta(\xi,E_2)).
\end{equation*}
It  suffices to prove that
\begin{equation*}
    \int_{\xi_0}^{\infty} \frac{\cos (2\theta(\xi,E_1)-2\theta(\xi,E_2))}{  1+\xi} d\xi=\frac{O(1)}{\xi_0^{\frac{1}{3}}},\int_{\xi_0}^{\infty} \frac{\cos (2\theta(\xi,E_1)+2\theta(\xi,E_2))}{  1+\xi} d\xi=\frac{O(1)}{\xi_0^{\frac{1}{3}}}.
 \end{equation*}
 By \eqref{GPrufeTmar14}, one has
 \begin{equation}\label{GPrufeTmar132}
  \frac{d(\theta(\xi,E_1)+\theta(\xi,E_2))}{d\xi}=2-Q(\xi,E_1)\sin^2\theta(\xi,E_1)-Q(\xi,E_2)\sin^2\theta(\xi,E_2).
\end{equation}
By \eqref{GPQ1}, \eqref{Gtheta1} and \eqref{GPrufeTmar132}, we have
\begin{equation*}
  \int_{\xi_0}^{\infty} \frac{\cos (2\theta(\xi,E_1)+2\theta(\xi,E_2))}{  1+\xi} d\xi=\frac{O(1)}{\xi_0^{\frac{1}{3}}}.
\end{equation*}
Thus
we only need to prove
\begin{equation}\label{Gkeyformar131}
    \int_{\xi_0}^{\infty} \frac{\cos (2\theta(\xi,E_1)-2\theta(\xi,E_2))}{  1+\xi} d\xi=\frac{O(1)}{\xi_0^{\frac{1}{3}}}.
 \end{equation}
By \eqref{GPrufeTmar14}  again, one has
\begin{eqnarray}
   \frac{d(\theta(\xi,E_1)-\theta(\xi,E_2))}{d\xi} &=& (-\frac{5}{36\xi^2}-V(\xi))\sin^2\theta(\xi,E_2)-
   (-\frac{5}{36\xi^2}-V(\xi))\sin^2\theta(\xi,E_1)\nonumber\\
   &&+\frac{E_1}{c\xi^{\frac{2}{3}}}\sin^2\theta(\xi,E_1)-\frac{E_2}{c\xi^{\frac{2}{3}}}\sin^2\theta(\xi,E_2) \nonumber\\
   &=& (-\frac{5}{36\xi^2}-V(\xi))\sin^2\theta(\xi,E_2)-
   (-\frac{5}{36\xi^2}-V(\xi))\sin^2\theta(\xi,E_1) \nonumber\\
   &&-
  \frac{1}{2} \frac{E_1}{c\xi^{\frac{2}{3}}}\cos2\theta(\xi,E_1)+\frac{1}{2}\frac{E_2}{c\xi^{\frac{2}{3}}}\cos2\theta(\xi,E_2)
    + \frac{E_1-E_2}{2c\xi^{\frac{2}{3}}}.
\end{eqnarray}
 Define
 \begin{equation*}
   \beta(\xi)=\frac{E_1}{2c\xi^{\frac{2}{3}}}\cos2\theta(\xi,E_1)-\frac{E_2}{2c\xi^{\frac{2}{3}}}\cos2\theta(\xi,E_2).
 \end{equation*}
Let $f(1)=\theta(1,E_1)-\theta(1,E_2)$ and
\begin{equation*}
 \frac{d f(\xi)}{d\xi}=(-\frac{5}{36\xi^2}-V(\xi))\sin^2\theta(\xi,E_2)-
   (-\frac{5}{36\xi^2}-V(\xi))\sin^2\theta(\xi,E_1)+
   \frac{E_1-E_2}{2c\xi^{\frac{2}{3}}}.
\end{equation*}
Thus
\begin{equation*}
  f(\xi)-(\theta(\xi,E_1)-\theta(\xi,E_2))=\int_1^{\xi} \beta(x)dx.
\end{equation*}
 By  \eqref{Gtheta1}, we have   for some $\beta_0$,
 \begin{equation*}
   \int_1^{\infty} \beta(x)dx=\beta_0,  \int_{\xi}^{\infty} \beta(x)dx=\frac{O(1)}{1+\xi^{\frac{1}{3}}}
 \end{equation*}
 and then
 \begin{equation*}
   \int_1^{\xi} \beta(x)dx=\beta_0+\frac{O(1)}{1+\xi^{ \frac{1}{3}}}.
 \end{equation*}
 Thus
 \begin{equation*}
  \theta(\xi,E_1)-\theta(\xi,E_2)=f(\xi)+\frac{O(1)}{1+\xi^{ \frac{1}{3}}}-\beta_0.
 \end{equation*}
 In order to prove \eqref{Gkeyformar131}, it suffices to prove that
 \begin{equation}\label{Gkeyformar132}
    \int_{\xi_0}^{\infty} \frac{\cos 2f(\xi)}{  1+\xi} d\xi=\frac{O(1)}{\xi_0^{\frac{1}{3}}},\int_{\xi_0}^{\infty} \frac{\sin 2f(\xi)}{  1+\xi} d\xi=\frac{O(1)}{\xi_0^{\frac{1}{3}}}.
 \end{equation}
By  changing the variable $y=\xi^{\frac{1}{3}}$, one has
 \begin{equation}\label{Gmar15}
   \frac{d f(y)}{dy}=\frac{df}{d\xi}\frac{d\xi}{dy}=\frac{3}{2c}(E_1-E_2) +\frac{O(1)}{1+y}.
 \end{equation}
 By  \eqref{Gtheta1} and \eqref{Gmar15}, we have
 \begin{equation*}
  \int_{\xi_0^{\frac{1}{3}}}^{\infty} \frac{\cos 2f(y)}{  1+y} dy=\frac{O(1)}{\xi_0^{\frac{1}{3}}},\int_{{\xi_0^{\frac{1}{3}}}}^{\infty} \frac{\sin 2f(y)}{  1+y} dy=\frac{O(1)}{\xi_0^{\frac{1}{3}}}.
 \end{equation*}
 This implies  \eqref{Gkeyformar132}. We finish the proof.
  \end{proof}
\section{ Single embedded eigenvalue}\label{SE}

\begin{proof}[\bf Proof of Theorem \ref{Mainthm1}]\label{SE}
As mentioned before, we only consider $\alpha=1$.
By the assumption of Theorem \ref{Mainthm1},
 one has
 \begin{equation*}
   \limsup_{\xi\to \infty}\frac{1}{2}\xi |\frac{q(c\xi^\frac{2}{3})} {c\xi^{\frac{2}{3}}}|=d<\frac{\pi}{12}.
 \end{equation*}
Let $\epsilon$ be a small positive number. Then there exists some $\xi_0>0$ such that for all $\xi>\xi_0$,
\begin{equation*}
    \frac{1}{2}\xi |\frac{q(c\xi^\frac{2}{3})} {c\xi^{\frac{2}{3}}}|<d+\epsilon<\frac{\pi}{12}.
\end{equation*}

 By \eqref{GPrufRmar14} and Lemma \ref{Keyle1} ($a=\xi_0$ and $b=\xi$), one has for large $\xi_0$ and $\xi>\xi_0$,
 \begin{eqnarray*}
    \log R(\xi,E)-\log R(\xi_0,E) &\geq&  O(1)-(d+\epsilon)\int_{\xi_0}^{\xi}\frac{|\sin2\theta(t,E)|}{t}dt \\
    &\geq & O(1)- \frac{2}{\pi}(d+\epsilon)\ln \xi.
 \end{eqnarray*}
 Thus
 \begin{equation}\label{Gmar311N}
  R(\xi,E) \geq \frac{1}{C \xi^{ \frac{2}{\pi}(d+\epsilon)}}
 \end{equation}
 for large $\xi$.
Let us estimate the $L^2(\R^+)$ norm of $R(\xi,E)$. Direct computation implies that ($\epsilon$ is small enough)
\begin{eqnarray*}
    R^2(\xi,E)p(\xi) &\geq& \frac{1}{C \xi^{\frac{2}{3}}  \xi^{2 \frac{2}{\pi}(d+\epsilon)}} \\
  &\geq &  \frac{1}{C \xi}.
\end{eqnarray*}
Thus $ R(\xi,E)\notin L^2(\R ^+,p(\xi)d\xi)$. This contradicts   $ \phi\in L^2(\R ^+,p(\xi)d\xi)$ and then contradicts   $u\in L^2(\R^+)$.
\end{proof}

\begin{proof}[\bf Proof of Theorem   \ref{Mainthm2} for non-critical points]
Fix any $E\in\R$.
In the case of $H=-D^2+ {v}+q$, let $u$ be the solution of $Hu=Eu$ with the  boundary condition of $H$ at $x=0$.
In the case of $\widetilde{H}=-D^2+\widetilde{v}+q$, let $q=0$ for $x<0$ and let  $u$ be the solution of $\widetilde{H}u=Eu$ such that $u$ satisfies  \eqref{Gnesc1}.
So $u\in L^2(-\infty,0]$.

We employ the same notations in the proof of Theorem \ref{Mainthm1}.
 We define $q(x)$  for $x>0$ by
 \begin{equation}\label{Defapr20}
    \frac{1}{2}\frac{q(c\xi^{\frac{2}{3}})}{c\xi^{\frac{2}{3}}}=-\frac{d}{  {\xi}}{\rm sgn}(\sin2\theta(\xi,E)),
 \end{equation}
 where ${\rm sgn}(\cdot) $ is the sign function and $d>\frac{\pi}{12}$ is a constant, which will be determined later.

 Substitute  \eqref{Defapr20} into \eqref{GPrufeTmar14}, and
 solve the   nonlinear system for $\theta$ with the some proper boundary  condition $\theta(1,E)=\theta_0$.
It is not difficult to  see that \eqref{GPrufeTmar14}   has a unique piecewise
smooth global solution by a standard ODE existence and uniqueness theorem.

Thus $q$ is well defined and
\begin{equation}\label{e6}
    \frac{d\log R(\xi,E)}{d\xi}= (- \frac{5}{72\xi^2}-\frac{E}{2c\xi^{\frac{2}{3}}})\sin2\theta(\xi,E)-  d \frac{|\sin2\theta(\xi,E)|}{\xi}.
\end{equation}
By \eqref{e6} and Lemma \ref{Keyle1} ($a=\xi_0$ and $b=\xi$), one has for large $\xi_0$ and $\xi>\xi_0$,
 \begin{eqnarray*}
    \log R(\xi,E)-\log R(\xi_0,E) &\leq&  O(1)-d\int_{\xi_0}^{\xi}\frac{|\sin2\theta(t,E)|}{t}dt \\
    &\leq & O(1)-\frac{2}{\pi}d\ln \xi.
 \end{eqnarray*}
 Thus
 \begin{equation}\label{Gmar311}
  R(\xi,E) \leq \frac{C}{ \xi^{\frac{2}{\pi}d}}
 \end{equation}
 for large $\xi$.
 Thus for some small $\epsilon>0$,
\begin{eqnarray*}
    R^2(\xi,E)p(\xi) &\leq& \frac{C}{\xi^{\frac{2}{3}}  \xi^{ \frac{4}{\pi}d}} \\
  &\leq &  \frac{C}{ \xi^{1+\epsilon}},
\end{eqnarray*}
since $d>\frac{\pi}{12}$. This implies  $ R(\xi,E)\in L^2(\R ^+,p(\xi)d\xi)$ and then  $u\in L^2(\R^+)$.
In the case of $\widetilde{H}$, we also have $u\in L^2(\R)$.

For any $a >\frac{\pi}{4}$ in the Theorem \ref{Mainthm2}, let $d=\frac{a}{3}$.  By the definition of \eqref{Defapr20},  we have
\begin{equation*}
   \limsup_{x\to \infty}\sqrt{{x}}|q(x)|=a.
 \end{equation*}
We finish the proof.
\end{proof}
\begin{proof}[\bf Proof of Theorem \ref{Mainthm2} for  the critical point ]
In this case $a=\frac{\pi}{4}$. Let $d_0=\frac{\pi}{12}$.
We employ the same notations in the proof of non-critical case.
Let $\epsilon_n=\frac{1}{n}$ and  $a_n= 2^{\frac{4}{\pi}n^3}$.

 We define $q(x)$  for $x>0$ piecewisely.
 For $\xi\in[a_n,a_{n+1})$, we define
 \begin{equation}\label{Def}
    \frac{1}{2}\frac{q(c\xi^{\frac{2}{3}})}{c\xi^{\frac{2}{3}}}=-\frac{d_0+\epsilon_n}{  {\xi}}{\rm sgn}(\sin2\theta(\xi,E)).
 \end{equation}

Suppose $ q(c\xi)$ is defined for $\xi\in(0,a_n]$. Let $\theta_n=\theta(a_n,E).$
 Substitute  \eqref{Def} into \eqref{GPrufeTmar14}, and
 solve the   nonlinear system for $\theta$ with  the  boundary  condition $\theta(a_n,E)=\theta_n$.

Thus we have for $a_n\leq\xi\leq a_{n+1}$,
\begin{equation}\label{e6apr20}
    \frac{d\log R(\xi,E)}{d\xi}= (- \frac{5}{72\xi^2}-\frac{E}{2c\xi^{\frac{2}{3}}})\sin2\theta(\xi,E)-  (d_0+\epsilon_n) \frac{|\sin2\theta(\xi,E)|}{\xi}.
\end{equation}

By \eqref{e6apr20} and Lemma \ref{Keyle1} ($a=a_n$ and $b=a_{n+1}$), one has
 \begin{eqnarray*}
    \log R(a_{n+1},E)-\log R(a_n,E) &\leq& \frac{O(1)}{a_n^{\frac{1}{3}}}-\frac{2}{\pi}(d_0+\epsilon_n)\ln a_{n+1},
 \end{eqnarray*}
 and for $a_n\leq \xi \leq a_{n+1}$,
 \begin{equation*}
    \log R(\xi,E)-\log R(a_n,E) \leq  \frac{O(1)}{a_n^{\frac{1}{3}}}-\frac{2}{\pi}(d_0+\epsilon_n)\ln \xi.
 \end{equation*}
Thus, one has
\begin{equation*}
  R(a_n,E)=O(1).
\end{equation*}
 Moreover, for $a_n\leq \xi\leq a_{n+1}$, one has
\begin{eqnarray*}
    R^2(\xi,E)p(\xi) &\leq&  O(1)R^2(a_n,E)\frac{1}{\xi^{\frac{2}{3}}  \xi^{ \frac{4}{\pi}(d_0+\epsilon_n)}} \\
  &\leq &   \frac{O(1)}{ \xi^{1+\frac{4}{\pi}\epsilon_n}}.
\end{eqnarray*}
Direct computation shows that
\begin{equation*}
 \int_{a_n}^{a_{n+1}} R^2(\xi,E)p(\xi)d\xi\leq \frac{O(1)}{\epsilon_na_n^{\frac{\pi}{4}\epsilon_n}}=O(1)\frac{n}{2^{n^2}}.
\end{equation*}
 This implies  $ R(\xi,E)\in L^2(\R ^+,p(\xi)d\xi)$ and then  $u\in L^2(\R^+)$ ($u\in L^2(\R)$).
\end{proof}
\section{ Proof of Theorem \ref{Maintheoremapr3}}\label{finitelymany}
%

  \begin{lemma}\cite[Lemma 4.4]{kiselev1998modified}\label{Leapr7}
  Let $\{e_i\}_{i=1}^N$ be a set of unit vector in a Hilbert space $\mathcal{H}$  so that
  \begin{equation*}
    \alpha=N\sup_{j\neq k}| \langle e_k,e_j\rangle|<1.
  \end{equation*}
  Then
  \begin{equation}\label{Gapr71}
    \sum_{i=1}^N|\langle g,e_i\rangle|^2\leq (1+\alpha)||g||^2.
  \end{equation}
  \end{lemma}
  \begin{proof}[\bf Proof of Theorem \ref{Maintheoremapr3}]
Let
   \begin{equation*}
 V(\xi)= \frac{q(c\xi^{\frac{2}{3}})}{c\xi^{\frac{2}{3}}}.
\end{equation*}
By the assumption of Theorem \ref{Maintheoremapr3}, for any $M>\frac{2}{3}a$, we have
\begin{equation*}
 | V(\xi)|\leq \frac{M}{1+\xi}
\end{equation*}
for large $\xi$.
By shifting the operator, we can assume
\begin{equation}\label{Gapr76}
  | V(\xi)|\leq \frac{M}{1+\xi}
\end{equation}
for all $\xi>0$.
  Suppose we have $N$ eigenvalues and denote them by $E_1,E_2,\cdots E_N$.
 It implies that  for $i=1,2,\cdots,N$,
  \begin{equation*}
    R(\xi,E_i)\in L^2(\R ^+,p(\xi)d\xi),
  \end{equation*}
  and  then
  \begin{equation*}
    \sum_{i=1}^NR(\xi,E_i)\in L^2(\R ^+,p(\xi)d\xi).
  \end{equation*}
  Thus there exists $B_j\to \infty$ such that
  \begin{equation*}
    R^2(B_j,E_i)p(B_j)\leq  \frac{1}{100}B_j^{-1},
  \end{equation*}
   and then by \eqref{GWei1}, one has
   \begin{equation}\label{Gapr72}
      R(B_j,E_i)\leq B_j^{-\frac{1}{6}},
   \end{equation}
   for all $i=1,2,\cdots,N$.

  By \eqref{Gapr72} and \eqref{GPrufRmar14}, we have
  \begin{equation}\label{Gapr73}
\int_{1}^{B_j} \frac{1}{2}Q(\xi,E_i)\sin2\theta(\xi,E_i) d\xi=  \int_{1}^{B_j}\frac{d}{d\xi} \log R(\xi,E_i)d\xi\leq -\frac{1}{6}\log B_j+O(1).
  \end{equation}
By Lemma \ref{Keyle1}, one has
\begin{equation}\label{Gapr74}
  \int_{1}^{B_j}(-\frac{5}{36\xi^2}-\frac{E}{c\xi^{\frac{2}{3}}})\sin2\theta(\xi,E_i)d\xi=O(1).
\end{equation}
By \eqref{Gapr73} and \eqref{Gapr74}, we have
\begin{equation}\label{Gapr75}
  \int_{1}^{B_j}V(\xi)\sin2\theta(\xi,E_i) d\xi \leq -\frac{1}{3}\log B_j+O(1).
\end{equation}
Now consider the Hilbert spaces
\begin{equation*}
  \mathcal{H}_j=L^2([1,B_j],(1+\xi)d\xi).
\end{equation*}
In  $\mathcal{H}_j$, by \eqref{Gapr76} we have
\begin{equation}\label{Gapr77}
  ||V||_{ \mathcal{H}_j}^2\leq M^2\log (1+B_j).
\end{equation}
Let
\begin{equation*}
  e^j_{i}(\xi)=\frac{1}{\sqrt{A_i^j}}\frac{\sin 2\theta(\xi,E_i)}{1+\xi}\chi_{[1,B_j]}(\xi),
\end{equation*}
where $A_i^j$ is chosen such that $e_i^j$ is a unit vector in $\mathcal{H}_j$.
We have the following estimate,
\begin{eqnarray}
  A_i^j &=& \int_1^{B_j}\frac{\sin^2 2\theta(\xi,E_i)}{1+\xi}d\xi \nonumber\\
   &=& \int_1^{B_j}\frac{1}{2(1+\xi)}d\xi- \int_1^{B_j}\frac{\cos  4\theta(\xi,E_i)}{2(1+\xi)}d\xi\nonumber\\
   &=& \frac{1}{2}\log B_j+O(1),\label{Gapr79}
\end{eqnarray}
since $\int_1^{B_j}\frac{\cos  4\theta(\xi,E_i)}{2(1+\xi)}d\xi=O(1)$ by  Lemma \ref{Keyle1}.

By Lemma \ref{Keyleboud}, we have for $i\neq k$,
\begin{equation*}
  \int_1^{B_j}  \frac{\sin  2\theta(\xi,E_i)\sin  2\theta(\xi,E_k)}{1+\xi}d\xi=O(1).
\end{equation*}
It yields that
\begin{equation}\label{Gapr78}
  \langle e_i^j,e_k^j \rangle=\frac{O(1)}{\log B_j}.
\end{equation}
By \eqref{Gapr79} and \eqref{Gapr75}
\begin{equation}\label{Gapr710}
 \langle V,e^j_i \rangle_{\mathcal{H}_j}\leq -\frac{\sqrt{2}}{3}\sqrt{\log B_j}+O(1).
\end{equation}
By \eqref{Gapr71} and \eqref{Gapr78}, one has
\begin{equation}\label{Gapr711}
\sum_{i=1}^N  |\langle V,e^j_i\rangle_{\mathcal{H}_j}|^2\leq (1+\frac{O(1)}{\log B_j})||V||_{\mathcal{H}_j}.
\end{equation}
By \eqref{Gapr710} and \eqref{Gapr77},
we have
\begin{equation*}
  N\frac{2}{9}\log B_j\leq M^2 \log B_j+O(1).
\end{equation*}
Let $j\to \infty$, we get
\begin{equation*}
  N \leq \frac{9}{2}M^2,
\end{equation*}
for any $M>\frac{2}{3}a$. This implies
\begin{equation*}
  N \leq 2a^2.
\end{equation*}

  \end{proof}

\section{Effective single piece  constructions}\label{ESPC}

For the case of $\widetilde{H}$, we let $q(x)=0$ for $x<0$. In both cases, let
 \begin{equation}\label{Gv}
  \frac{q(c\xi^{\frac{2}{3}})}{c\xi^{\frac{2}{3}}}=V(\xi), \text{  for }\xi>0.
\end{equation}
Our goal is to construct $V(\xi)\approx \frac{1}{1+\xi}$ and then get   $q(x)\approx\frac{1}{1+\sqrt{x}}$ by solving \eqref{Gv}.

Denote  by
\begin{equation}\label{GPQ1mar13}
 Q(\xi,E)= -\frac{5}{36\xi^2}-\frac{E}{c\xi^{\frac{2}{3}}}+V(\xi).
\end{equation}
Suppose we construct function $q$ on $[0,x]$. We can define $u$ on $[0,x]$.
Let $u$  be  the  solution of $Hu=Eu$  on $[0,x]$ with some boundary condition at $0$. 

Under the  Liouville transformation,
$\phi$ satisfies
\begin{equation}\label{Gschximar13}
  -\phi^{\prime\prime}+Q(\xi,E)\phi=\phi.
\end{equation}

Recall that we have
\begin{equation}\label{GPrufR}
 \frac{d\log R(\xi,E)}{d\xi}=\frac{1}{2}Q(\xi,E)\sin2\theta(\xi,E)
\end{equation}
and
\begin{equation}\label{GPrufeTmar13}
  \frac{d\theta(\xi,E)}{d\xi}=1-Q(\xi,E)\sin^2\theta(\xi,E).
\end{equation}

\begin{theorem}\label{Twocase}
Fix $M>0$.
Let  $E\in \R$ and $ A=\{{E}_j\}_{j=1}^N$. Suppose $E\notin A$ and $\{E_j\}_{j=1}^N$ are distinct.
Suppose  $\theta_0\in[0,\pi]$. Let $\xi_1>\xi_0>b$.
Then there exist constant  $C(E, A)$ (independent of $b, \xi_0$ and $\xi_1$) and potential $ V(M,\xi,E,A,\xi_0,\xi_1,b,\theta_0)$  such that the following holds:

   \begin{description}
     \item[Potential]   for $\xi_0\leq \xi \leq \xi_1$, ${\rm supp}( V)\subset(\xi_0,\xi_1)$,  $  V\in C^{\infty}(\xi_0,\xi_1)$,  and
     \begin{equation}\label{thm141}
        | V(M,\xi,E,A,\xi_0,\xi_1,b,\theta_0)|\leq \frac{4M}{\xi-b}
     \end{equation}

     \item[Solution for $E$]Let $Q(\xi,E)$ be given by \eqref{GPQ1mar13}. Then the solution of $(-D^2+Q(\xi,E))\phi=\phi$ with boundary condition $\frac{\phi^\prime(\xi_0)}{\phi(\xi_0)}=\tan\theta_0$ satisfies
     \begin{equation}\label{thm142}
        R(\xi_1,E)\leq (1+ \frac{CM}{(\xi_0-b)^{\frac{1}{3}}})(\frac{\xi_1-b}{\xi_0-b})^{-M} R(\xi_0,E)
     \end{equation}
     and  for $\xi_0<\xi<\xi_1$,
      \begin{equation}\label{thm143}
        R(\xi,E)\leq  (1+ \frac{CM}{(\xi_0-b)^{\frac{1}{3}}})R(\xi_0,E).
     \end{equation}
      \item[Solution for ${E}_j$] Let $Q(\xi,E_j)$ be given by \eqref{GPQ1mar13}. Then the solution of $(-D^2+ Q(\xi,E_j))\phi=\phi$ with any boundary condition at $\xi_0$ satisfies
      for $\xi_0<\xi\leq \xi_1$,
      \begin{equation}\label{thm144}
        R(\xi,{E}_j)\leq  (1+ \frac{CM}{(\xi_0-b)^{\frac{1}{3}}}) R(\xi_0,{E}_j).
     \end{equation}
   \end{description}

  \end{theorem}
  \begin{proof}
  By changing $\xi$ to $\xi-b$, we assume $b=0$.
  We consider the non-linear differential equation for $\xi>0$,
 \begin{equation}\label{Gnonlinear}
   \frac{d \theta(\xi,E,\xi_0,\theta_0)}{d\xi}=1-(-\frac{5}{36\xi^2}-\frac{E}{c\xi^{\frac{2}{3}}}-\frac{4M}{  1+\xi}\sin2\theta)\sin^2\theta,
 \end{equation}
 where $C$ is a large constant that will be chosen later. Solving  \eqref{Gnonlinear} on $[\xi_0,\infty)$ with
 initial condition $\frac{ \theta^{\prime}(\xi_0)}{\theta(\xi_0)}=\tan \theta_0$, we get a unique solution.
 Let
 \begin{equation*}
   V(\xi)=-\frac{4M}{  1+\xi}\sin2\theta(\xi,E,\xi_0,\theta_0).
 \end{equation*}
 Let $Q(\xi,E)$ be given by \eqref{GPQ1mar13}.

 We will prove that $V$ satisfies  Theorem \ref{Twocase} under some modifications.

 The solution $\phi(\xi,E)$ of \eqref{Gschximar13} satisfies
 \begin{eqnarray}
    \log R(\xi,E)- \log R(\xi_0,E) &=& -\int_{\xi_0}^{\xi} \frac{2M}{ 1+x}\sin^22\theta(x,E) dx+ \int_{\xi_0}^{\xi} (- \frac{5}{72x^2}-\frac{E}{2cx^{\frac{2}{3}}})\sin2\theta(x,E) dx\nonumber \\
    &=& -\int_{\xi_0}^{\xi}\frac{M}{   1+x}dx+ \int_{\xi_0}^{\xi} \frac{M}{ 1+x}\cos4\theta(x,E)dx\nonumber \\
     && +\int_{\xi_0}^{\xi} ( -\frac{5}{72x^2}-\frac{E}{2cx^{\frac{2}{3}}})\sin2\theta(x,E) dx.\label{GPrufRmar13}
 \end{eqnarray}

 By \eqref{Gtheta1} and \eqref{GPrufeTmar13}, one has
 \begin{equation*}
   \int_{\xi_0}^{\xi} \frac{ 1}{ 1+x}\cos4\theta(x,E)dx=\frac{O(1)}{\xi_0}, \int_{\xi_0}^{\xi} (- \frac{5}{36x^2}-\frac{E}{cx^{\frac{2}{3}}})\sin2\theta(x,E) dx=\frac{O(1)}{\xi_0^{\frac{1}{3}}}.
 \end{equation*}
 By \eqref{GPrufRmar13}, we prove \eqref{thm142} and \eqref{thm143}.

 Let us move to the proof of \eqref{thm144}.
 The solution $\phi(\xi,E_j)$   satisfies
 \begin{eqnarray}
    \log R(\xi,E_j)- \log R(\xi_0,E_j) &=& \frac{1}{2}\int_{\xi_0}^{\xi} (-\frac{5}{36x^2}-\frac{E_j}{cx^{\frac{2}{3}}}-\frac{4M}{  1+x}\sin2\theta(x,E))
 \sin2\theta(x,E_j) dx\nonumber \\
    &=&  - \int_{\xi_0}^{\xi} \frac{2M}{ 1+x} \sin2\theta(x,E)
 \sin2\theta(x,E_j) dx\nonumber \\
  && +\frac{1}{2}\int_{\xi_0}^{\xi} (- \frac{5}{36x^2}-\frac{E_j}{cx^{\frac{2}{3}}})\sin2\theta(x,E_j) dx\label{GPrufRmar131}
 \end{eqnarray}
and
\begin{equation}\label{GPrufeTmar131}
  \frac{d\theta(\xi,E_j)}{d\xi}=1-Q(\xi,E_j)\sin^2\theta(\xi,E_j).
\end{equation}

 By \eqref{Gtheta1} and \eqref{GPrufeTmar131}, one has
 \begin{equation*}
     \int_{\xi_0}^{\xi} ( -\frac{5}{36x^2}-\frac{E_j}{cx^{\frac{2}{3}}})\sin2\theta(x,E_j) dx=\frac{O(1)}{\xi_0^{\frac{1}{3}}}.
 \end{equation*}
 Thus in order to prove   proof of \eqref{thm144}, we only need to prove
 \begin{equation}\label{Gkeyformar13}
    \int_{\xi_0}^{\xi} \frac{\sin2\theta(x,E)
 \sin2\theta(x,E_j) }{  1+x} dx=\frac{O(1)}{\xi_0^{\frac{1}{3}}}.
 \end{equation}
 By   Lemma \ref{Keyleboud},  \eqref{Gkeyformar13} is true since $E\notin A$.
 Thus $V$ satisfies the construction except the regularity.
 This can be done if we modify $V$ a little at two end points $\xi_0$ and $\xi_1$, and keep all the bounds.
  \end{proof}
\begin{remark}
We can also obtain the explicit  formula  for $C$ which depends on $|E|$, $|E_j|$ and $|E-E_j|$, $j=1,2,3,\cdots,N$ (see \cite{liuwkb,liu2018absence}). However, those constant $C$ only change the $L^2 $ norms by a factor, which does not influence the bounds of the  potentials.
\end{remark}
\section{ Universal gluing constructions}\label{UGC}
Let $\{E_j\}_{j=1}^N$ be any $N$ different points in $\R$.  We will use the piecewise functions to complete our constructions.
Let $B=\{E_j\}_{j=1}^N$ and $$S=100\max_{E_j\in B} \{C(E_j, B\backslash E_j)\},$$
where $C$ is given by Theorem \ref{Twocase}.

Let $T_w$ be a sequence and $J_w=1+N\sum_{j=1}^w T_j$, where  $w\in \Z^+$. Let $M>0$ and $J_0=1$.
Let  function $q(x)=0$ for $x\in[0,c]$ (that is $V(\xi)=0$ for $\xi\in[0,1]$). We can define $u$ on $[0,c]$ as following.
Let $u$  be  the  solution of $Hu=Eu$  on $[0,c]$ with the boundary condition at $0$. 
Thus we can define $\phi(\xi,E)$ for $\xi\in(0,1]$.

Now we  will   define function $V$ ($\text{supp} V\subset (1,\infty)$ ) and $\phi(\xi,E_j)$, $j=1,2,\ldots N$ on $(1,J_w)$ by induction, such that
\begin{enumerate}[1.]
\item
$\phi(\xi,E_j)$ solves  for $\xi\in (0,J_w)$
\begin{align}\label{eigenengj}
     \left( -\frac{d^2}{d\xi^2} +V(\xi)\right) \phi(\xi,E_j)
     =E_j  \phi(\xi,E_j),
\end{align}
and satisfies  boundary condition
\begin{equation}\label{1boundaryn}
    \frac{\phi^{\prime}(\frac{1}{2},E_j)}{\phi(\frac{1}{2},E_j)}=\tan\theta_j,
\end{equation}
\item
$\phi(\xi,E_j)$  for $j=1,2,\cdots,N$ and $w\geq 1$, satisfies
\begin{equation}\label{eigenjapr}
 R(J_{w},E_j) \leq (1+\frac{SM}{\sqrt[3]{J_{w-1}}})^{N} (\frac{J_{w-1}+T_{w}}{J_{w-1}})^{-M} R(J_{w-1},E_j),
\end{equation}
and  also for $\xi\in[J_{w-1},J_w]$
\begin{equation}\label{eigenj}
 R(\xi,E_j) \leq (1+\frac{SM}{\sqrt[3]{J_{w-1}}})^{N}  R(J_{w-1},E_j).
\end{equation}
\item  $V(\xi)\in C^{\infty}(0,J_w]$  and for $\xi\in [J_{w-1},J_w]$, one has
\begin{equation}\label{controlkr}
    | V(\xi) |\leq (1+\frac{(N-1)T_{w}}{J_{w-1}})\frac{4M }{\xi}.
\end{equation}
\end{enumerate}


 We proceed by an induction argument. Suppose  we completed the construction  $V(\xi)$ for step $w$. Accordingly, we can define  $\phi(\xi,E_j)$ on $(1,J_w]$ for all $j=1,2,\cdots,N$ by \eqref{eigenengj} and \eqref{1boundaryn}.
Applying Theorem  \ref{Twocase} to $\xi_0=J_w$, $\xi_1=J_w+T_{w+1}$, $b=0$, $E=E_1$, $\tan\theta_0=\frac{\phi^\prime(J_w,E_1)}{\phi(J_w,E_1)}$ and $A= B\backslash \{E_1\}$,  we can define
$V(\xi,E_1,B\backslash \{E_1\},J_w,J_{w}+T_{w+1},0,\theta_0)$  on $\xi\in (J_w, J_w+T_{w+1}]$ and also $\phi(\xi,E_1)$ on $(J_w, J_w+T_{w+1}]$.
Since the boundary condition matches at the point $ J_w$ (guaranteed by $\tan\theta_0=\frac{\phi^\prime(J_w,E_1)}{\phi(J_w,E_1)}$), $\phi(\xi,E_1)$ is well defined on $(1,J_w+T_{w+1}]$  and satisfies  \eqref{eigenengj} and \eqref{1boundaryn}.   We  define $\phi(\xi,E_j) $ on $(0,J_{w}+T_{w+1})$  by  \eqref{eigenengj} and \eqref{1boundaryn} for all  $j=2,3,\cdots,N$. Thus  $\phi(\xi,E_j)$ is well defined on $(1,J_w+T_{w+1}]$, and satisfies  \eqref{eigenengj} and \eqref{1boundaryn} for all $j=1,2,\cdots,N$.
 Moreover,  letting $\xi_1=J_w+T_{w+1}$  in Theorem  \ref{Twocase}, one has (by \eqref{thm142})
 \begin{equation}\label{Gkstep1}
   R(J_w+T_{w+1},E_1) \leq (1+\frac{SM}{\sqrt[3]{J_w}})(\frac{J_w+T_{w+1}}{J_w})^{-M} R(J_w,E_1),
 \end{equation}
 and for all $\xi\in[J_w,J_w+T_{w+1}]$, we have (by \eqref{thm143})
 \begin{equation}\label{Gkstep1new}
   R(\xi,E_1) \leq (1+\frac{SM}{\sqrt[3]{J_w}}) R(J_w,E_1).
 \end{equation}
%



 Suppose we give the definition  of $V$ and $\phi(\xi,E_j)$ for all $j$ on $(0,J_w+tT_{w+1}]$ for $t\leq N-1$. Let us give the definition on $(0,J_w+(t+1)T_{w+1}]$.

 Applying  Theorem \ref{Twocase}  to $\xi_0=J_w+tT_{w+1}$, $\xi_1=J_w+(t+1)T_{w+1}$, $b=tT_{w+1}$, $E=E_{t+1}$, $A=B\backslash E_{t+1}$ and $\tan \theta_0=\frac{\phi^\prime(J_w+tT_{w+1},E_{t+1})}{\phi(J_w+tT_{w+1},E_{t+1})}$,  we can define
 ${V}(\xi,E_{t+1}, B_{w+1}\backslash E_{t+1}, J_w+tT_{w+1},J_w+(t+1)T_{w+1},tT_{w+1},\theta_0)$ on $\xi\in (J_w+tT_{w+1}, J_w+(t+1)T_{w+1})$. Similarly,  we can define $\phi(\xi,E_j) $ on $(0,J_{w}+(t+1)T_{w+1}]$  for all  $j=1,2,\cdots,N$.
 Moreover,  letting $\xi_1=J_w+(t+1)T_{w+1}$ in  Theorem  \ref{Twocase}, one has
 \begin{equation}\label{Gkstept}
   R(J_w+(t+1)T_{w+1},E_{t+1}) \leq (1+\frac{SM}{\sqrt[3]{J_w}})(\frac{J_w+T_{w+1}}{J_w})^{-M} R(J_w+tT_{w+1},E_{t+1}),
 \end{equation}
 and also for $\xi\in [J_w+tT_{w+1},J_w+(t+1)T_{w+1}]$, one has
 \begin{equation}\label{Gksteptnew}
   R(\xi,E_{t+1}) \leq (1+\frac{SM}{\sqrt[3]{J_w}}) R(J_w+tT_{w+1},E_{t+1}).
 \end{equation}
 By induction, we can define  $V(\xi)$ and $\phi(\xi,E_j)$  for all $j=1,2,\cdots,N$ on $(0,J_w+NT_{w+1}]=(0,J_{w+1}]$.

 Now we should show that the definition satisfies the $w+1$ step conditions \eqref{eigenengj}-\eqref{controlkr}.

 Let us  pick up  $R(\xi,E_j)$ for  some $E_j\in B$.
 $R(\xi,E_j)$ decreases  from   point  $J_w+(j-1)T_{w+1}$ to $J_w+jT_{w+1}$, and
  may increase from any point $J_w+(m-1)T_{w+1}$ to $J_w+mT_{w+1}$, $m=1,2,\cdots,N$ and $m\neq j$.
  That is
  \begin{equation*}
    R(J_w+jT_{w+1},E_j)\leq (1+\frac{SM}{\sqrt[3]{J_w}})(\frac{J_w+T_{w+1}}{J_w})^{-M} R(J_w+(j-1)T_{w+1},E_{j}),
  \end{equation*}
  and   for $m\neq j$,
  \begin{equation*}
    R(J_w+mT_{w+1},E_j)\leq (1+\frac{SM}{\sqrt[3]{J_w}}) R(J_w+(m-1)T_{w+1},E_{j}),
  \end{equation*}
  by  Theorem  \ref{Twocase}.

   Thus for $j=1,2,\cdots,N$,
   \begin{equation*}
    R(J_{w+1},E_j)\leq (1+\frac{SM}{\sqrt[3]{J_w}})^{N} (\frac{J_w+T_{w+1}}{J_w})^{-M} R(J_{w},E_j).
   \end{equation*}
   This leads to \eqref{eigenjapr}.
By the same arguments, we have
    for $j=1,2,\cdots,N$ and $\xi\in[J_w,J_{w+1}]$,
   \begin{equation*}
    R(\xi,E_j)\leq (1+\frac{SM}{\sqrt[3]{J_w}})^{N}  R(J_{w},E_j).
   \end{equation*}
This implies (\ref{eigenj}) for $w+1$.

 By the construction of $V(\xi)$,  we have
 for $\xi\in[J_w+tT_{w+1},J_w+(t+1)T_{w+1}]$ and $0\leq t\leq N-1$,
 \begin{equation}\label{b'1}
   |V(\xi)| \leq \frac{4M}{\xi-tT_{w+1}}.
 \end{equation}
In order to prove \eqref{controlkr}, it suffices to show for all $\xi\in[J_w+tT_{w+1},J_w+(t+1)T_{w+1}]$,
\begin{equation*}
  \frac{1}{\xi-tT_{w+1}}\leq (1+\frac{(N-1)T_{w+1}}{J_w}) \frac{1}{\xi}.
\end{equation*}
It suffices (we only need check  $\xi=J_w+tT_{w+1}$) to prove
\begin{equation*}
   \frac{1}{J_w}\leq (1+\frac{(N-1)T_{w+1}}{J_w}) \frac{1}{J_w+t T_{w+1}}.
\end{equation*}
Since $t\leq N-1$, we only need to show
\begin{equation*}
   \frac{1}{J_w}\leq  (1+\frac{(N-1)T_{w+1}}{J_w}) \frac{1}{J_w+(N-1) T_{w+1}},
\end{equation*}
which is true by calculations.
\section{Proof of Theorems of \ref{Mainthm3} and \ref{Mainthm4}, and all the Corollaries}\label{Twoapp}

\begin{proof}[\bf Proof of Theorems of \ref{Mainthm3}]
The case $N=1$ has been addressed in Theorem \ref{Mainthm2}.
Suppose $N\geq 2$.

 Let  $\epsilon=\frac{1}{\sqrt{\ln N}}$ and  $M=\frac{1}{6}+\frac{1}{6\epsilon}$.
 For $w\in \Z^+,$ let $T_w=N^{(1+\epsilon)w}$  so that
 $$J_w=1+N\sum_{i=1}^w N^i=1+N\frac{N^{(1+\epsilon)w+1}-1}{N^{(1+\epsilon)}-1}.$$

 It is easy to check that
 \begin{equation*}
   \lim_{w\to \infty}\frac{J_w+T_{w+1}}{J_w}=N^{\epsilon}-\frac{1}{N}+1.
 \end{equation*}
 For large $w$, say $w\geq w_0$, one has
 \begin{equation*}
   \frac{J_w+T_{w+1}}{J_w}\geq N^{\epsilon}+\frac{1}{4}.
 \end{equation*}
 Thus by \eqref{eigenjapr},  we have  for  $w\geq w_0$,
\begin{equation*}
 R(J_{w},E_j) \leq  (1+\frac{SM}{\sqrt[3]{J_w}})^{N}(N^{\epsilon}+\frac{1}{4})^{- M} R(J_{w-1},E_j),
\end{equation*}
and then
\begin{equation*}
 R(J_{w},E_j) \leq (1+\frac{SM}{\sqrt[3]{J_w}})^{N(w-w_0)}(N^{\epsilon}+\frac{1}{4})^{- M(w-w_0)} R(J_{w_0},E_j).
\end{equation*}
By \eqref{eigenj}, we have for $\xi\in[J_{w},J_{w+1}]$,
\begin{equation}\label{eigenjnew}
 R(\xi,E_j) \leq  (1+\frac{SM}{\sqrt[3]{J_w}})^{N(w+1-w_0)}(N^{\epsilon}+\frac{1}{4})^{-M(w-w_0)} R(J_{w_0},E_j).
\end{equation}
Let $\delta_w=\frac{1}{\sqrt{w}}$.
Let $p_{\delta_w}=3-\delta_w$. Let  $q_{\delta_w}$ be such that $\frac{1}{p_{\delta_w}}+\frac{1}{q_{\delta_w}}=1$. Then $q_{\delta_w}=\frac{3}{2}+\frac{\delta_w}{4-2\delta_w}$.

By basic inequality, one has
\begin{equation}\label{Gapr142}
  \int_{J_{w+1}}^{J_{w+2}} R^2(\xi,E_j)p(\xi) d\xi\leq (\int_{J_{w+1}}^{J_{w+2}} R^{2p_{\delta_w}}(\xi,E_j) d\xi)^{\frac{1}{p_{\delta_w}}}(\int_{J_{w+1}}^{J_{w+2}}p(\xi)^{q_{\delta_w}} d\xi)^{\frac{1}{q_{\delta_w}}}.
\end{equation}
Direct computations show that
\begin{eqnarray}
  \int_{J_{w+1}}^{J_{w+2}}p(\xi)^{q_{\delta_w}} d\xi &\leq & O(1)\frac{1}{\delta_w} \frac{1}{J_{w+1}^{\frac{1}{10}\delta_w}}\nonumber\\
   &=& O(1)  \frac{\sqrt{w}} { N^{\frac{1+\epsilon}{10} \sqrt{w} }} \nonumber\\
    &=& O(1).\label{Gapr151}
\end{eqnarray}
It is easy to see
\begin{equation}\label{Gapr152}
   (1+\frac{SM}{\sqrt[3]{J_{w+1}}})^{N2p_{\delta_w}w}=O(1).
\end{equation}
By \eqref{eigenjnew} and \eqref{Gapr152}, one has
\begin{equation*}
\int_{J_{w+1}}^{J_{w+2}} R^{2p_{\delta_w}}(\xi,E_j) d\xi\;\;\;\;\;\;\;\;\;\;\;\;\;\;\;\;\;\;\;\;\;\;\;\;\;\;\;\;\;\;\;\;\;\;\;\;\;\;\;\;\;\;\;\;\;\;\;\;\;\;\;\;\;\;\;\;\;\;\;\;\;\;\;\;\;\;\;\;\;\;\;\;\;\;\;\;\;\;\;\;\;\;\;\;\;\;\;\;
\;\;\;\;\;\;\;\;\;\;\;\;\;\;\;\;\;\;\;\;\;\;\;\;\;\;\;\;\;
\end{equation*}
\begin{eqnarray}
   &\leq& (1+\frac{SM}{\sqrt[3]{J_{w+1}}})^{N2p_{\delta_w}(w+2-w_0)}(N^{\epsilon}+\frac{1}{4})^{-2p_{\delta_w} M(w+1-w_0)} R(J_{w_0},E_j)^{2p_{\delta_w}}\int_{J_{w+1}}^{J_{w+2}}  d\xi\nonumber  \\
    &\leq&  O(1)(N^{\epsilon}+\frac{1}{4})^{-2p_{\delta_w}   M(w-w_0)} R(J_{w_0},E_j)^{2p_{\delta_w}} NT_{w+1} \nonumber  \\
    &\leq&  O(1)(N^{\epsilon}+\frac{1}{4})^{-2p_{\delta_w}   Mw}   N^{(1+\epsilon)w}\nonumber \\
     &=&  O(1)(1+ \frac{1}{4N^{\epsilon}})^{-2p_{\delta_w}   Mw}  N^{-2\epsilon p_{\delta_w}   Mw} N^{(1+\epsilon)w}\nonumber \\
      &=&  O(1)(1+ \frac{1}{4N^{\epsilon}})^{-2p_{\delta_w}   Mw}N^{2\epsilon  \epsilon_w   Mw} N^{-6\epsilon   Mw} N^{(1+\epsilon)w} \nonumber\\
        &=&  O(1)(1+ \frac{1}{4N^{\epsilon}})^{-2p_{\delta_w}   Mw}N^{2\epsilon  \epsilon_w   Mw} ,\label{Gapr153}
\end{eqnarray}
where the last equality holds by the fact $6\epsilon M=1+\epsilon$.

Direct computation  shows
\begin{eqnarray}
  \sum_{w=w_0}^{\infty}\left [(1+ \frac{1}{4N^{\epsilon}})^{-2p_{\delta_w}   Mw}N^{2\epsilon  \epsilon_w   Mw}\right]^{\frac{1}{p_{\delta_w}}} &=& \sum_{w=w_0}^{\infty} \left[(1+ \frac{1}{4N^{\epsilon}})^{-2p_{\delta_w}   Mw}N^{2\epsilon  M\sqrt{w}}\right]^{\frac{1}{p_{\delta_w}}} \nonumber \\
   &\leq &  \sum_{w=w_0}^{\infty}\left[(1+ \frac{1}{4N^{\epsilon}})^{-4  Mw}N^{2\epsilon  M\sqrt{w}}\right]^{\frac{1}{4}}\nonumber \\
    &\leq &\sum_{w=w_0}^{\infty}(1+ \frac{1}{4N^{\epsilon}})^{-  Mw}N^{\frac{\epsilon}{2}  M\sqrt{w}}\nonumber \\
    &<&\infty.\label{Gapr154}
\end{eqnarray}
By \eqref{Gapr142}, \eqref{Gapr151}, \eqref{Gapr153} and \eqref{Gapr154}, we have
$R(\xi,E_j)\in L^2(\xi, p(\xi)d\xi)$
for all $j=1,2,\cdots,N$.
This implies  $E_j$ is an  eigenvalue, $j=1,2,\cdots,N$.

By the fact $\lim _{w\to\infty}\frac{T_{w+1}}{J_w}=N^{\epsilon}-\frac{1}{N}$, and
\eqref{controlkr}, one has
 \begin{eqnarray}
    \limsup _{\xi\to \infty}\xi|V(\xi)|&\leq& 4(1+(N-1)(N^{\epsilon}-\frac{1}{N}))M \nonumber\\
     &=& \frac{2}{3} (1+(N-1)(N^{\epsilon}-\frac{1}{N})) (1+\frac{1}{\epsilon})\nonumber\\
     &=&  \frac{2}{3} N^{1+\epsilon} (1+\frac{1}{\epsilon})\nonumber\\
     &\leq&  \frac{2}{3} N^{1+\epsilon} e^{\frac{1}{\epsilon}}\nonumber\\
       &=& \frac{2}{3}e^{2\sqrt{\ln N}}N.
 \end{eqnarray}
 By \eqref{Gv}, we have
\begin{equation*}
  \limsup _{\xi\to \infty}\sqrt{\xi}|q(\xi)|\leq e^{2\sqrt{\ln N}}N.
\end{equation*}
We finish the proof.
\end{proof}

Let $M=100$  and $K=100CM$ in Theorem \ref{Twocase}. We get
\begin{proposition}\label{Twocase1}
Let  $E\in \R$ and $ A=\{{E}_j\}_{j=1}^k$. Suppose $E\notin A$ and $\{E_j\}_{j=1}^k$ are distinct.
Suppose  $\theta_0\in[0,\pi]$. Let $\xi_1>\xi_0>b$.
Then there exist constants $K(E, A)$, $C(E, A)$ (independent of $b, \xi_0$ and $\xi_1$) and potential $ V(\xi,E,A,\xi_0,\xi_1,b,\theta_0)$  such that  for $\xi_0-b>K(E,A)$ the following holds:

   \begin{description}
     \item[Potential]   for $\xi_0\leq \xi \leq \xi_1$, ${\rm supp}( V)\subset(\xi_0,\xi_1)$,  $  V\in C^{\infty}(\xi_0,\xi_1)$,  and
     \begin{equation*}
        | V(\xi,E,A,\xi_0,\xi_1,b,\theta_0)|\leq \frac{C(E, A)}{\xi-b}
     \end{equation*}

     \item[Solution for $E$]Let $Q(\xi,E)$ be given by \eqref{GPQ1mar13}. Then the solution of $(-D^2+Q(\xi,E))\phi=\phi$ with boundary condition $\frac{\phi^\prime(\xi_0)}{\phi(\xi_0)}=\tan\theta_0$ satisfies
     \begin{equation*}
        R(\xi_1,E)\leq 2(\frac{\xi_1-b}{\xi_0-b})^{-100} R(\xi_0,E)
     \end{equation*}
     and  for $\xi_0<\xi<\xi_1$,
      \begin{equation*}
        R(\xi,E)\leq 2R(\xi_0,E).
     \end{equation*}
      \item[Solution for ${E}_j$] Let $Q(\xi,E_j)$ be given by \eqref{GPQ1mar13}. Then the solution of $(-D^2+ Q(\xi,E_j))\phi=\phi$ with any boundary condition at $\xi_0$ satisfies
      for $\xi_0<\xi\leq \xi_1$,
      \begin{equation*}
        R(\xi,{E}_j)\leq  2R(\xi_0,{E}_j).
     \end{equation*}
   \end{description}

  \end{proposition}
\begin{proof}[\bf Proof of Theorem \ref{Mainthm4}]
Once we have  Proposition \ref{Twocase1}, we can prove
    Theorem    \ref{Mainthm4}   by the    arguments in \cite{ld1,jl}. We omit the details here.
  \end{proof}

\begin{proof}[\bf Proof of  all the Corollaries]
Corollaries \ref{cor1} and \ref{cor3} follow from Theorems \ref{Mainthm1} and \ref{Maintheoremapr3} respectively.
Let  $u(x,E) $ be  the  solution of
$\widetilde{H}u=Eu$  on $(-\infty,0]$ such that $u\in L^2(-\infty,0]$ (this can be guaranteed by Lemma \ref{Keyle2}).
Let $\theta_E=\frac{u^{\prime}(0,E)}{u(0,E)}$. Instead of using boundary condition $\theta $ in the previous arguments,
we use $\theta_E$. Now    Corollaries \ref{cor2}, \ref{cor4} and \ref{cor5} follow from Theorems \ref{Mainthm2}, \ref{Mainthm3} and \ref{Mainthm4}.
\end{proof}
\section{Proof of general cases}\label{General}
In this section, we  will  adapt our proof for $\alpha=1$  to  general  $\alpha\in(0,2)$.
Take $v(x)=v_{\alpha}(x)=x^{\alpha}$ in \eqref{Gsch}-\eqref{GPQ}  and let  $c_{\alpha}=(1+\frac{\alpha}{2})^{\frac{2}{2+\alpha}}$.
 In the general cases, we have
 \begin{equation}\label{GLiou1new}
  x=c_{\alpha}\xi^{\frac{2}{2+\alpha}}, \phi_{\alpha}(\xi,E)=c_{\alpha}^{\frac{\alpha}{4}}\xi^{\frac{\alpha}{2(2+\alpha)}}u(c_{\alpha}\xi^{\frac{2}{2+\alpha}}),
 \end{equation}
 \begin{equation}\label{GWei1new}
  p_{\alpha}(\xi)= \frac{1}{c_{\alpha}^{\alpha}\xi^{\frac{2\alpha}{2+\alpha}}},
\end{equation}
and
\begin{equation}\label{GPQ1oldnew}
 Q_{\alpha}(\xi,E)= -\frac{5}{4}\frac{\alpha^2}{(2+\alpha)^2}\frac{1}{\xi^2}+\frac{\alpha(\alpha-1)}{(2+\alpha)^2}\frac{1}{\xi^{2}}+\frac{q(c_{\alpha}\xi^{\frac{2}{2+\alpha}})-E}{c_{\alpha}^{\alpha}\xi^{\frac{2\alpha}{2+\alpha}}}.
\end{equation}
Let
\begin{equation}\label{Gvaprnew}
 V_{\alpha}(\xi)=\frac{q(c_{\alpha}\xi^{\frac{2}{2+\alpha}})}{c_{\alpha}^{\alpha}\xi^{\frac{2\alpha}{2+\alpha}}}.
\end{equation}
Then
\begin{eqnarray*}
   Q_{\alpha}(\xi,E) &=& -\frac{5}{4}\frac{\alpha^2}{(2+\alpha)^2}\frac{1}{\xi^2}+\frac{\alpha(\alpha-1)}{(2+\alpha)^2}\frac{1}{\xi^{2}}-
   \frac{E}{c_{\alpha}^{\alpha}\xi^{\frac{2\alpha}{2+\alpha}}}+V_{\alpha}(\xi).
\end{eqnarray*}

Suppose $u\in L^2(\R^+)$  is a solution of \eqref{Gsch} with $v(x)=x^{\alpha}$. It follows that
$\phi_{\alpha}$ satisfies
\begin{equation}\label{Gschxinew}
  -\frac{d^2\phi_{\alpha}}{d\xi^2}+Q_{\alpha}(\xi,E)\phi_{\alpha}=\phi_{\alpha}.
\end{equation}

Now,   all the quantities such as  $Q_{\alpha}(\xi,E)$ and $c_{\alpha}$ depend on $\alpha$.

In order to proceed the proof in a similar way,   two important components are essential: 1. all the estimates of oscillated integral still hold;  2. how
the $\alpha$-dependent constants are computed in the main theorems.

 It is convenient to employ a slight  different Prf\"ufer transformation. Such standard  trick  has been   used to deal with Stark operators before (p.10 in \cite{christ2003absolutely}).

Let $H_{\alpha}(\xi,E)=-
   \frac{E}{c_{\alpha}^{\alpha}\xi^{\frac{2\alpha}{2+\alpha}}}$.
 The new  Pr\"{u}fer tranformation is given by
\begin{equation}\label{GPruf1new}
 \sqrt{1-H_{\alpha}(\xi,E)}\phi_{\alpha}(\xi,E)=R_{\alpha}(\xi,E)\sin\theta_{\alpha}(\xi,E),
\end{equation}
and
\begin{equation}\label{GPrufnew}
\frac{d \phi_{\alpha}(\xi,E)}{d \xi}=R_{\alpha}(\xi,E)\cos\theta_{\alpha}(\xi,E).
\end{equation}
By \eqref{Gschxinew}, we have
\begin{equation}\label{GPrufRmar14new}
 \frac{d\log R_{\alpha}(\xi,E)}{d\xi}=\frac{1}{2}\frac{V_{\alpha}(\xi)}{\sqrt{1-H_{\alpha}(\xi,E)}}\sin2\theta_{\alpha}(\xi,E)+ O(\frac{1}{\xi^{1+\frac{2\alpha}{2+\alpha}}})
\end{equation}
and
\begin{equation}\label{GPrufeTmar14new}
  \frac{d\theta_{\alpha}(\xi,E)}{d\xi}=\sqrt{1-H_{\alpha}(\xi,E)}-\frac{V_{\alpha}(\xi)}{\sqrt{1-H_{\alpha}(\xi,E)}}\sin^2\theta_{\alpha}(\xi,E)+
  O(\frac{1}{\xi^{1+\frac{2\alpha}{2+\alpha}}}).
\end{equation}
 As in this paper $ |V_{\alpha}(\xi)|\leq \frac{h(\xi)}{1+\xi}$,  for any $h(\xi)$ with $h(\xi)\to \infty$ as $\xi\to \infty$, one has
\begin{equation*}
  V_{\alpha}(\xi)=\frac{O(1)}{\xi^{1-\frac{\alpha}{2+\alpha}}}.
\end{equation*}
Finally, \eqref{GPrufRmar14new} and \eqref{GPrufeTmar14new}
become
\begin{equation}\label{GPrufRmar14new1}
 \frac{d\log R_{\alpha}(\xi,E)}{d\xi}=\frac{1}{2} V_{\alpha}(\xi) \sin2\theta_{\alpha}(\xi,E)+ O(\frac{1}{\xi^{1+\frac{\alpha}{2+\alpha}}}),
\end{equation}
and
\begin{equation}\label{GPrufeTmar14new1}
  \frac{d\theta_{\alpha}(\xi,E)}{d\xi}=\sqrt{1-H_{\alpha}(\xi,E)}- V_{\alpha}(\xi) \sin^2\theta_{\alpha}(\xi,E)+
  O(\frac{1}{\xi^{1+\frac{\alpha}{2+\alpha}}}).
\end{equation}
Since the tail $O(\frac{1}{\xi^{1+\frac{\alpha}{2+\alpha}}})$ in \eqref{GPrufRmar14new1} and \eqref{GPrufeTmar14new1} is integrable, it does not change the estimate  at all up to a constant.  Under the new Pr\"ufer transformation and following the  arguments  of proof of $\alpha=1$,   it is not hard to verify all the estimates of oscillated integral still hold.

Now we are going to convince the readers the $\alpha$-dependent constants in the main theorems.  There are two $\alpha$-dependent constants: the power decay  rate of $\xi$  (this constant is fixed in every theorem and it is equal $1-\frac{\alpha}{2}$)  and the constants in front of $x^{1-\frac{\alpha}{2}}$.

By \eqref{GPrufRmar14new1} and the proof of $\alpha=1$, the critical case of $V_{\alpha}(\xi)$ is $\frac{O(1)}{\xi}$. By the relation \eqref{Gvaprnew}, it is easy to see that
$ \frac{O(1)}{x^{1-\frac{\alpha}{2}}}$ is the critical case for $q(x)$. It explains where the constant $1-\frac{\alpha}{2}$ is from.

Now we are in the position to explain the constants in front of $x^{1-\frac{\alpha}{2}}$. Since the constants  are different in different theorems, we need to check them one by one.

Let us check the constant $\frac{2-\alpha}{4} \pi $ in Theorems \ref{Mainthm1} and \ref{Mainthm2} first.
Here are the details.
Suppose $q(x)\approx\frac{A}{ x^{1-\frac{\alpha}{2}}}$.
By \eqref{Gvaprnew}, $ V_{\alpha}(\xi)\approx \frac{A}{ \frac{1+\alpha}{2}}\frac{1}{\xi}$. Similar to \eqref{Gmar311N} and  \eqref{Gmar311},
\begin{equation}\label{Nov111}
 R^2(\xi,E)\approx\xi^{-\frac{4A}{\pi(2+ \alpha)}}.
\end{equation}

 Since the critical case of  $R^2_{\alpha}(\xi,E)p_{\alpha}(\xi)$ is $\frac{O(1)}{\xi}$ ($\frac{1}{\xi^{1+\delta}}$  is integrable and $\frac{1}{\xi^{1-\delta}}$ is not integrable for $\delta>0$), we have that the critical case for $R^2_{\alpha}(\xi,E)$ is
 \begin{equation}\label{Nov110}
 R^2_{\alpha}(\xi,E)  \approx \frac{O(1)}{p_{\alpha}(\xi)\xi}= O(1) \xi^{-\frac{2-\alpha}{2+\alpha}} .
 \end{equation}
 By \eqref{Nov111} and \eqref{Nov110}, we have $A=\frac{2-\alpha}{4}\pi.$ We finish  checking in Theorems \ref{Mainthm1} and \ref{Mainthm2}.

  For the rest of constants in main theorems, we can do the similar check.  Actually, from  \eqref{Nov111} and \eqref{Nov110}, we can see that only ratio $\frac{A}{2-\alpha}$ matters. This fact also holds for the rest of the theorems. By the fact that the ratio $\frac{A}{2-\alpha}$  does not depend on $\alpha$ and the constants for $\alpha=1$, we can easily see the constants in all the theorems are correct.

\appendix

\section{Proof of Lemma  \ref{Keyle2}}

By symmetry, it suffices to prove the following Lemma.

\begin{lemma}\label{Keyleapp1}
Suppose $\lim _{x\to \infty}\widetilde{q}(x)=\infty$. Let us consider equation
\begin{equation}\label{Gnescapp}
  -y^{\prime\prime}+\widetilde{q}(x)y=0.
\end{equation}
Then for any $M>0$,  there is a solution of \eqref{Gnescapp} such that
\begin{equation}\label{Gnesc1app}
  |y(x)| \leq  e^{-M|x|}
\end{equation}
for large $x$.
\end{lemma}
We give a Lemma first.
\begin{lemma}[Lemma C, \cite{kato}]\label{Keyleaap2}
Suppose $f(x)$ and $g(x)$ are continuous on $\R^+$, and $f(x)\leq g(x)$.
Consider the two differential equations
\begin{equation*}
  y^{\prime\prime}-f(x)y=0,  z^{\prime\prime}-g(x)z=0.
\end{equation*}
If the first equation has a positive solution $y_0(x)$, then the second equation has
two positive solutions $z_1(x)$ and $z_2(x)$ such that $\frac{z_1(x)}{y_0(x)}$ is non-increasing and
$\frac{z_2(x)}{y_0(x)}$ is non-deceasing.
\end{lemma}

\begin{proof}[\bf Proof of Lemma \ref{Keyleapp1}]
By the definition, there exists $x_0$ such that  $ \widetilde{q}(x)>4M^2+1$ for $x>x_0.$

Let $y_0=e^{-2Mx}$. Obviously,
\begin{equation*}
 y_0 ^{\prime\prime} - 4M^2y_0(x)=0.
\end{equation*}
Applying Lemma \ref{Keyleaap2} and the fact $ \widetilde{q}(x)>4M^2+1$ for $x>x_0$, one has
equation  \eqref{Gnescapp} has a positive solution $y$ such that
$\frac{y}{y_0}$ is non-increasing on $[x_0,\infty)$. This implies  \eqref{Gnesc1app}.


\end{proof}

 \section*{Acknowledgments}
 The author  wishes  to thank Peter Hislop  for some helpful conversations.
   The research was supported by  NSF DMS-1401204 and  NSF DMS-1700314.

\end{document}